\documentclass{ecai}
\usepackage{times}
\usepackage{graphicx}
\usepackage{latexsym}
\usepackage{float}
\usepackage{wrapfig}
\usepackage{amsmath}
\usepackage{diagbox}
\usepackage{framed}
\usepackage{enumerate}
\usepackage{subfigure}
\usepackage{amsmath}
\usepackage{booktabs}
\usepackage{enumerate}
\usepackage{framed}

\usepackage{amsfonts}
\usepackage{amssymb}
\usepackage{mathrsfs}
\usepackage{amsthm}
\usepackage{color}
\usepackage[numbers]{natbib}
\newtheorem{myDef}{Definition}

\newtheorem{proposition}{Proposition}

\newtheorem{thm}{Theorem}
\newtheorem{example}{Example}
                  
\begin{document}

\title{Maximal Information Propagation with Budgets}

\author{Haomin Shi \and Yao Zhang \and Zilin Si \and Letong Wang \and Dengji Zhao\institute{ShanghaiTech University, China, email: \{shihm, zhangyao1, sizl, wanglt, zhaodj\}@shanghaitech.edu.cn} }

\maketitle
\bibliographystyle{ecai}

\begin{abstract}
In this paper, we present an information propagation game on a network where the information is originated from a sponsor who is willing to pay a fixed total budget to the players who propagate the information. Our solution can be applied to real world situations such as advertising via social networks with limited budgets.
The goal is to design a mechanism to distribute the budget such that all players in the social network are incentivized to propagate information to all their neighbours.
We propose a family of mechanisms to achieve the goal, where propagating information to all neighbours is a dominant strategy for all players. Furthermore, we also consider the cases where the budget has to be completely shared.
\end{abstract}


\section{Introduction}


In recent years, with the help of the internet, people start taking advantage of social media to propagate information to more people through their friends~\cite{doi:10.1080/15252019.2011.10722189,arney2013networks}. For instance, if an online retailer wants to advertise her products but cannot afford a high advertising fee on ad platforms, one alternative way is to incentivize her customers to propagate the information via their friends. As another example, if a survey team needs many people to fill their questionnaires, they can incentivize people to spread the survey further. The drive of all the applications is propagating information via social networks~\cite{miller2010social}. The challenge is how to incentivize people to do so.

There exists a rich literature that is based on the social network to study classical marketing or mechanism design problems. For example, Brill et al.~\cite{brill2016false} proposed a false-name-proof information collection algorithm for recommendation decisions in social networks, in which they introduced the concept of weight on each node in the social network. Shen et al.~\cite{shen2019multi} proposed a multi-winner contests (MWC) mechanism for strategic contest information diffusion on social networks, which satisfied several desirable properties including false-name-proof, individual rationality, budget constraint, monotonicity, and subgraph constraint. For marketing, Emek et al.~\cite{emek2011mechanisms} presented a theoretical framework to design reward mechanisms for multi-level marketing within social networks and established a full characterization of the geometric mechanism family. Later, Drucker and Fleischer~\cite{drucker2012simpler} proposed a family of mechanisms to prevent Sybil attacks by profiting maximizers while preserving its original properties in multi-level marketing. Moreover, Li et al.~\cite{li2017mechanism} proposed an auction mechanism in a single-item auction setting to incentivize buyers to report their valuations to the seller truthfully as well as propagate the auction information to all their neighbors on a social network and Zhao et al.~\cite{zhao2018selling} further generalized it to the setting of multiple items. For cooperative games, Douceur and Moscibroda~\cite{douceur2007lottery} proposed a lottery tree that motivates people to join or contribute to a network system without participation benefit, but proportionally rewards each participant on its contribution and their leading-in participants' contributions. Furthermore, Abbassi and Misra~\cite{abbassi2011multi} proposed a multi-level revenue sharing for referral-based viral marketing over online social networks to achieve compatibility, individual rationality and potential reach.

There also exist some works on finding and propagating influential nodes in a social network in order to trigger a cascade of further participation. Wang et al.~\cite{7018025} proposed a price-performance-ratio inspired heuristic scheme to economically select seeds within a given budget to maximize the diffusion process. David et al.~\cite{kempe2003maximizing} provided provable approximation guarantees for efficient algorithms of selecting the most influential nodes. Galeotti and Goyal~\cite{galeotti2009influencing} investigated how firms can harness the power of social networks to promote their sales and profits.

There are also some works on incentive referring and participation. Naroditskiy et al.~\cite{naroditskiy2014referral} conducted a field experiment that compared several mechanisms for incentivizing propagation via social media with the support of a charitable cause. Gao et al.~\cite{gao2015survey} surveyed the literature of theoretical frameworks and experimental studies of the incentive strategies used in participatory sensing.

All the works above share the same motivation of propagating information via social networks. In this paper, we focus on how to incentivize people to propagate some information from a sponsor, which is a high-level abstraction of the above settings. In our setting, a sponsor provides a fixed budget to reward the people who propagate information for her. The goal is to incentivize participants to propagate information to all their social neighbours. Note that this is not achievable if we simply use a fixed reward for each invitation because the sponsor will certainly spend more than the budget in this case.

Under our circumstance, the key point is that the budget distributed to all agents is bounded at the beginning. Hence, inviting more people means there are more people to share the limited reward. The well-known winning solution from the 2009 DARPA red balloon challenge is a successful attempt~\cite{pickard2011time}, but it only works on a tree model when there is an actual task to do while in our setting, it is just information propagation. It also cannot fully use the budget in our setting. To combat these problems, we propose a novel reward sharing scheme that incentives people to propagate the information to all their neighbours while the reward shared among them can completely uses the budget.

Specifically, our contributions advance the state of the art as follows:
\begin{itemize}
    \item We define an information propagation game on a graph/network and define the concept of \textit{propagation incentive compatible} to guarantee the information will be fully propagated in the graph.
    \item We propose a novel reward distribution scheme that satisfies the new concept 
    to achieve maximal information propagation. Moreover, it spends all the budget, i.e., it is budget balanced.
    \item Some practical issues have also been discussed to show the applicability of our proposed mechanisms.
\end{itemize}

The rest of the paper is presented as follows. Section 2 introduces our model and the preliminaries. We first propose a simple mechanism in Section 3 to better understand the challenge. Then, in Section 4 we propose our novel budget distribution scheme that guarantees maximal information propagation. Finally, we discuss some practical issues in Section 5, and conclude in Section 6.





\section{The Model}
We focus on an information propagation game where a sponsor aims to propagate a piece of information to as many agents as possible via agents' social connections. In the real-world applications, the connections between agents can represent friendship or neighborhood, and the information could be an advertisement or any other information. We investigate a reward mechanism for the sponsor under a fixed budget to incentivize agents to propagate the information to all their neighbours.

Formally, let $N = \{S,1,2,\dots,n\}$ be the set of all agents in the network and a special agent $S \in N$, which is the sponsor, intends to propagate a piece of information to the others. We model the network as an unweighted directed acyclic graph (DAG) $G=(N, E)$ with source $S$. Each edge $(i,j)\in E$ indicates that agent $i$ can propagate the information to agent $j$. 
All the rewards given by the sponsor to the others come from a fixed budget $\mathcal{B}>0$. Let $\mathcal{G}$ be the space of all networks satisfying our setting. We define the reward mechanism as follows.

\begin{myDef}
A \textbf{reward mechanism} is defined by $M$: $\mathcal{G} \times \mathbb{R}_+ \rightarrow \mathbb{R}^{|N|}$. Given a graph $G = (N,E) \in \mathcal{G}$ and a budget $\mathcal{B}\in \mathbb{R}_+$, the output of the mechanism is $M(G,\mathcal{B}) = \mathbf{r} = (r_S, r_1, r_2, \dots, r_n)$, where $r_i$ is the reward assigned to agent $i$ for $i\in \{1,2,\dots,n\}$ and $r_S$ is the remaining budget that has not been distributed.
\end{myDef}

One minimum requirement of the reward mechanism is that the sum of reward $r_{i}$ equals the budget $\mathcal{B}$.

\begin{myDef}
A mechanism $M$ is \textbf{feasible} if for all $G= (N,E)\in \mathcal{G}$ and all $\mathcal{B}\in \mathbb{R}_+$, the output $M(G,\mathcal{B}) = \mathbf{r} = (r_S, r_1, r_2, \dots, r_n)$ satisfies $\sum_{i\in N} r_i = \mathcal{B}$.
\end{myDef}

From the sponsor's interests, the rewards assigned by the mechanism should be bounded by the budget. Otherwise, the sponsor will have a deficit which is not allowed in our model. This is called weakly budget balanced.

\begin{myDef}
\label{def:bb}
A feasible mechanism $M$ is \textbf{weakly budget balanced} (WBB) if for all $G= (N,E)\in \mathcal{G}$ and all $\mathcal{B}\in \mathbb{R}_+$, the output $M(G,\mathcal{B}) = \mathbf{r}$ satisfies $r_S \geq 0$. More strictly, a mechanism is \textbf{budget balanced} (BB) if $r_S=0$.
\end{myDef}

Moreover, to make the reward mechanism applicable, agents should not suffer a loss from participating. We call this individual rationality.

\begin{myDef}
A feasible mechanism $M$ is \textbf{individually rational} (IR) if and only if for all $G= (N,E)\in \mathcal{G}$ and all $\mathcal{B}\in \mathbb{R}_+$, the output $M(G,B) = \mathbf{r} = (r_S, r_1, r_2, \dots, r_n)$ satisfies $\forall i\in N$, $r_i \geq 0$.
\end{myDef}

The last and the most important property we investigate in this paper is incentivizing information propagation in the network. This is defined as propagation incentive compatibility.

\begin{myDef}
A feasible mechanism $M$ is \textbf{propagation incentive compatible} (PIC) if for all $G = (N,E)\in \mathcal{G}$, all $\mathcal{B}\in \mathbb{R}_+$, all $i\in N\setminus \{S\}$ and all edge set $e \subseteq \{(x,y)| (x,y) \in E, x= i\}$, we have
\[ M(G, \mathcal{B})_i \geq M(G'_e, \mathcal{B})_i \]
where $G'_e$ is the connected graph containing $S$ deduced from $G$ without the edge set $e$.

We say $M$ is \textbf{strongly propagation incentive compatible} (SPIC) if the strict inequality holds.
\end{myDef}

Propagation incentive compatibility indicates that each agent at least does not suffer a loss from the information propagation. If a mechanism is strongly propagation incentive compatible, propagating information to all her neighbours is all agents' unique dominant strategy. 

Given all the above definitions, we first show that trivial mechanisms cannot satisfy them in our model. For example, simply give each agent a fixed reward for participation or uniformly divide the budget among all participants. The former is not budget-balanced, and the later does not satisfy PIC. This is summarized in the following propositions.

\begin{proposition} If a feasible mechanism $M$ gives each agent a fixed reward $r \in \mathbb{R}_+$ for all agents participated, then $M$ is not weakly budget balanced.
\end{proposition}
\begin{proof}
For a fixed budget $\mathcal{B}$, if the size of agents $n = \lceil \mathcal{B}/r \rceil + 1$, then which $ n \cdot r > \mathcal{B}$ and $r_S<0$. Hence, the property of weakly budget balance cannot be satisfied.
\end{proof}

\begin{proposition} If a feasible mechanism $M$ uniformly divides the budget among all agents participated, then $M$ is not PIC.
\end{proposition}
\begin{proof}
For a fixed budget $\mathcal{B}$ and $N\setminus\{S\}$ of size $n$, each agent gets $\frac{\mathcal{B}}{n}$ under $M$. If one of the agents propagates fewer agents, then the size of $N'\setminus\{S\}$ will be $n' \leq n$, which suggests that her new reward will be $\frac{\mathcal{B}}{n'} \geq \frac{\mathcal{B}}{n}$. Hence, $M$ is not PIC.
\end{proof}




\section{A Starter Mechanism}
\label{sec:simple}
First, we propose a simple reward mechanism that is weakly budget balanced and propagation incentive compatible. The starter mechanism
helps us understand the challenges we face and motivate our advanced mechanisms in the later sections.

Let us start with some extra notations. Denote the shortest distance of an agent $i$ to the sponsor by $d_i$, and the set of all agents with depth $l$ by $L_l = \{ i|d_i = l \}$, which is also called the layer $l$. Having defined each node's layer, we only keep those edges from upper layers to lower layers. Our starter mechanism first distributes the total reward to different layers. For all agent in layer $l$, the reward $B_l$, which is distributed among $L_l$, is defined as
\[ B_l = \beta^{l-1} (1-\beta) \cdot \mathcal{B} \]
where $\beta$ is a discount factor which satisfies $0<\beta<1$.

It is clear that the total reward to each layer is fixed so that inviting more agents does not decrease the inviter's reward since the invitees are in the next layer and won't join the current layer's reward distribution. Therefore, for all agents in $L_l$, we can simply divide the budget $B_l$ proportionally to them according to the number of agents invited by each agent. More specifically, let $n_i$ be the number of agents whom agent $i$ has informed, and $f:\mathbb{R}\rightarrow \mathbb{R}$ be a positive and monotone increasing function. Then the reward distributed to agent $i\in L_l$ is defined as
\[ r_i = \frac{f(n_i)}{\sum_{j\in L_l}f(n_j)} \cdot B_l \]

The full description of the above mechanism is illustrated below.
 
\begin{framed}
 \noindent\textbf{The Starter Mechanism}
 
 \noindent\rule{\textwidth}{0.5pt}
 
 \noindent\textsc{Input:} the graph $G$ and the budget $\mathcal{B}$.
 
 \begin{enumerate}
     \item Using breadth first search to compute the agent sets in each layer $L_1, L_2,\dots, L_{l_{\text{max}}}$.
     \item For each layer $l\in\{1,\dots,l_{\text{max}}\}$, let $B_l = \beta^{l-1} (1-\beta) \cdot \mathcal{B}$.
     \item For each agent $i\in L_l$, let $r_i = \frac{f(n_i)}{\sum_{j\in L_l}f(n_j)} \cdot B_l$.
 \end{enumerate}
 
 \noindent\textsc{Output:} the reward vector $\mathbf{r}$.
\end{framed}
 

 
 
 
    
    
    
    
        
    

Example 1 will illustrate how the starter mechanism works.

\begin{example}
Consider the following information propagation network illustrated in Figure~\ref{fig:algorithm}. The sponsor $S$ firstly propagates the information to A and B. Then A and B further propagate the information to their friends $C$, $D$ and $F$, and so on.

\begin{figure}[H]
\centering
\includegraphics[width = 8cm]{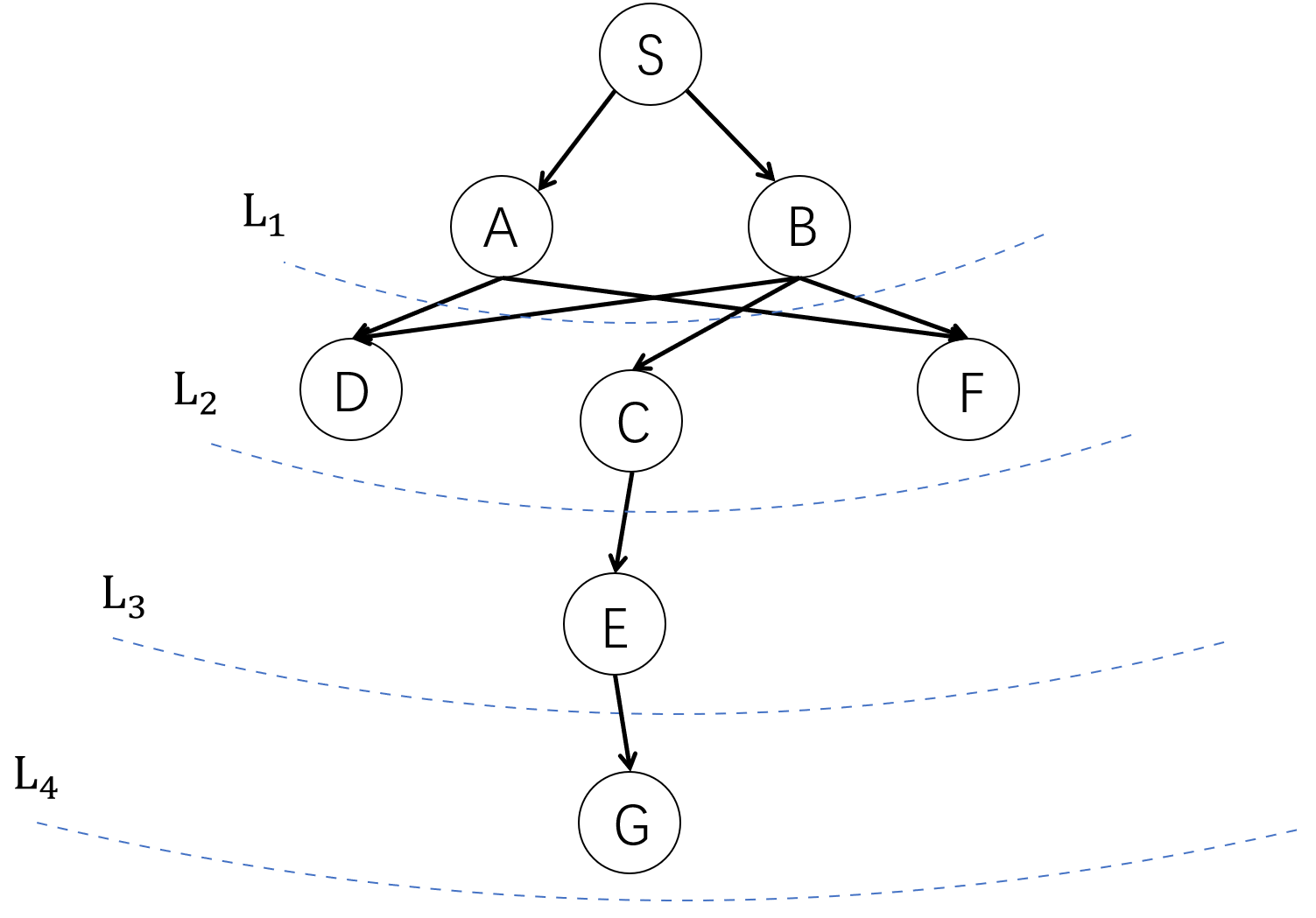}
\caption{An example of the information propagation network.}
\label{fig:algorithm}
\end{figure}

Suppose the total budget from $S$ is $\mathcal{B} = 10$, we apply the starter mechanism and let $\beta = \frac{1}{2}$, $f(n)=n$.


%

Taking the first two layers as example to compute the rewards, for $L_1$, the reward to A and B are
\begin{gather*}
    r_A = B_1 \cdot \frac{n_A}{n_A+n_B} = \frac{1}{2}\mathcal{B} \cdot \frac{n_A}{n_A+n_B} = 2 \\
    r_B = B_1 \cdot \frac{n_B}{n_A+n_B} = \frac{1}{2}\mathcal{B} \cdot \frac{n_B}{n_A+n_B} = 3
\end{gather*}

For $L_2$, only agent $C$ propagates information to others. Hence, agent $C$ will take all the reward distributed to $L_2$, i.e.,
 \begin{equation*}
     r_C = B_{l=2} = \frac{1}{4}\mathcal{B} = 2.5
 \end{equation*}
and $r_D = r_F = 0$, because $n_D = n_F = 0$.
\end{example}

Proposition~\ref{prop:1bb} shows that the starter mechanism has some nice properties.

\begin{proposition}\label{prop:1bb}
The starter mechanism is weakly budget balanced, individual rational and propagation incentive compatible.
\end{proposition}

\begin{proof}
For weakly budget balanced, check the total reward distributed to all agents except the sponsor:
\begin{equation*}
    \sum_{l=1}^{l_{\text{max}}} B_{l} \leq \sum_{l=1}^{\infty} B_{l} = (1-\beta)\cdot \frac{1}{1-\beta}\cdot \mathcal{B} = \mathcal{B}
\end{equation*}
then $r_S = \mathcal{B}-\sum_{l=1}^{l_{\text{max}}} B_{l} \geq 0$ is directly hold.

For individual rationality, since $0<\beta<1$, $\mathcal{B}>0$ and $f$ is a positive function, then for every agent $i\in L_l$,
\[ r_i = \frac{f(n_i)}{\sum_{j\in L_l}f(n_j)} \cdot B_l = \frac{f(n_i)}{\sum_{j\in L_l}f(n_j)} \cdot \beta^{l-1} (1-\beta) \cdot \mathcal{B} \geq 0 \]

For propagation incentive compatibility, since $f$ is monotone increasing, for every agent $i\in L_l$, if she does not propagate information to all her neighbours, which is $n_i'<n_i$, we have
\[ r_i' = \frac{f(n_i')}{f(n_i')+\sum_{j\in L_l,j\neq i}f(n_j)} \leq \frac{f(n_i)}{\sum_{j\in L_l}f(n_j)} = r_i \]
where $r_i'$ is the reward if $i$ only propagates information to $n_i'$ neighbours\footnote{Note that the inequality will be strictly hold if $f$ is strictly monotone increasing. In this case, the mechanism will be strongly PIC.}.

Therefore the mechanism is weakly budget balanced, individual rational and propagation incentive compatible.

\end{proof}



\section{A General Scheme}
Although the starter mechanism described in Section~\ref{sec:simple} has satisfied the properties of WBB, IR and PIC, there are still some concerns. The first concern is that the budget cannot be fully used because the graph is limited. In another word, $l_{max}$ is finite. This may cause a problem when using up the resource is a critical requirement. For example, in virtual cash market  such as bitcoins, the amount of bitcoins mined everyday is a fixed number, and if it is used as rewards, it has to be fully distributed. Also, when a government wants to distribute some social welfare, usually it wishes to fully distribute it.

Another concern is about ``distinguishability": once the discount factor $\beta$ is fixed, the reward allocated to each layer is fixed, no matter how many agents they have informed to and how network is structured like in the immediate layer. The starter mechanism cannot distinguish them, which means that it may not strongly incentivize agents to invite more participants.


To tackle the problems of the starter mechanism, we first further characterize the starter mechanism by considering how we share the rewards between two adjacent layers, $L_l$ and $L_{l+1}$. Let $B_l$ be the total budget passed to layer $L_l$ before they shared with $L_{l+1}$, and $b_i'$ be the current reward for agent $i$ in $L_l$. Thus, we have $B_l' = \sum_{i\in L_l} b_i'$. Then the reward are shared between $L_l$ and $L_{l+1}$ and the final reward left for $L_l$ is defined by
\begin{equation}
    B_l = \frac{\beta^{l-1}(1-\beta)}{\beta^{l-2}}B_l' = \beta(1-\beta)B_l' = (1-\beta')B_l'
    \label{eq:beta'}
\end{equation}
where $\beta'=\beta^2-\beta+1$ and $0<\beta<1$.

Therefore, the reward shared to $L_{l+1}$ is $\beta'B'_l$. Or alternatively,
\[ B_l = (1-\beta')\sum_{i\in L_l}b_i',\quad B'_{l+1} = \beta'\sum_{i\in L_l}b_i' \]

Suppose for each agent $j$ in $L_{l+1}$, we distribute $b_j'$ as the current reward from $\beta'B'_l$ and this will drive the mechanism between the layer $l+1$ and $l+2$, and so on so forth.

At last, we distribute the reward $B_l$ to agents in each layer $L_l$. This form of distribution is more precise and can be used as a basis for further generalization.

\subsection{Budget Distribution Scheme}
The problem of budget balance and ``distinguishability" in the starter mechanism occurs because the proportion of the budget distributed to the agents in $L_l$ is independent of the number of agents in each layer. To tackle this problem, we introduce a new term when we consider the reward sharing between layers $L_l$ and $L_{l+1}$ as
\[ B_{l} = (1-\beta)\sum_{i\in L_l} b_i'+\beta\sum_{i\in L_l} A_i b_i' \]
where $\beta$ is a discount factor and $A_i$ is a bonus factor such that $0<\beta <1$ and $0\leq A_i\leq1$.
Here we assume that there should be more than one node in the first layer, i.e. the sponsor should know more than one agent in the social network. Otherwise, the single agent will always have no incentive to propagate since she can take the whole budget if no other agents have been informed.
Intuitively, the  bonus factor here is to provide extra reward for the propagation effect of the agents in the layer. Now the generalized budget distribution scheme is given as below.

\begin{framed}
 \noindent\textbf{Budget Distribution Scheme}
 
 \noindent\rule{\textwidth}{0.5pt}
 
 \noindent\textsc{Input:} the graph $G$ and the budget $\mathcal{B}$.
 
 \begin{enumerate}
     \item Using breadth first search to compute the layer sets $L_1, L_2,\dots, L_{l_{\text{max}}}$ and $L_{l_{\text{max}+1}}$.
     \item For each $i\in L_1$, set $b_i' = \mathcal{B}/|L_1|$.
     \item For each $l$ in $\{1,\dots,l_{\text{max}}\}$
     \begin{enumerate}
         \item For each $i\in L_l$ compute $A_i$ according to its propagation.
         \item Let $B_{l} = (1-\beta)\sum_{i\in L_{l}} b_i'+\beta\sum_{i\in L_{l}} A_i b_i'$ and $B_{l+1}' = \sum_{i\in L_{l}} b_i' - B_l$.
         \item Distribute $B_l$ to agents in $L_l$, i.e., for agent $i$ in $L_l$, she gets $r_i$ as reward.
         \item Distribute $B_{l+1}'$ to agents in $L_{l+1}$, i.e., for agent $j$ in $L_{l+1}$, she gets $b_j'$ as current reward.
     \end{enumerate}
 \end{enumerate}
 
 \noindent\textsc{Output:} the reward vector $\mathbf{r}$.
\end{framed}


 
 
 
 
    
    
    
    
    

Notice that the starter mechanism is a special case of the budget distribution scheme when the bonus factor $A_i$ is always 0 and $\beta$ in the scheme serves as the role of $\beta'$ in Equation~\eqref{eq:beta'}. We can easily show that all mechanisms under this scheme are IR and WBB.

\begin{thm}\label{thm:bds}
The budget distribution scheme is IR and WBB.
\end{thm}

\begin{proof}
For IR, notice that $0<\beta < 1$ and $0\leq A_i\leq1$. Then for any layer $L_l$, we have
\begin{align*}
    (1-\beta)\sum_{i\in L_{l}} b_i'& \leq B_l \leq (1-\beta)\sum_{i\in L_{l}} b_i'+\beta\sum_{i\in L_{l}} b_i'\\
    0& < B_l<B_l'
\end{align*}

For WBB, notice that the total reward distributed to agents is at most $B_1 + B_2' = B_1' = \mathcal{B}$.

Hence, the properties of IR and WBB holds.
\end{proof}



Intuitively, to distinguish the differences of the networks, 
we give an extra bonus for their propagation. The bonus is additional to the proportion in the starter mechanism, which is the key to maintain the property of PIC, because the fixed proportion guarantees that new comers will not decrease their inviters' rewards.

Now we consider how to choose a proper bonus factor $A_i$. To incentivize propagation, we design $A_i$ to be a function depending on the number of agents been propagated. Considering an agent $i\in L_l$, let $n_i$ be the number of agents $i$ directly propagates to and $n_{-i}$ be the number of agents who all others in $L_l$ propagates to, i.e., $n_{-i}=\sum_{j\in L_l j\neq i}n_j$.

In order to give agents in $L_l$ more rewards if they propagate the information to more neighbours, $A_i$ should increase when $n_i$ increases. On the other hand, $A_i$ should decrease when $n_{-i}$ increases on the purpose of bounding the bonus given to agents in $L_l$ (to avoid the budget left to the next layer being too less) and creating competition among agents in layer $L_l$ (where the difference between agents' propagation also matters).
Therefore, we introduce a reasonable design of $A_i$ as:

\[ A_i = (1-\alpha)^{n_{-i}} \]
where $0<\alpha<1$ is the division factor. Note that under this definition, $A_i$ is decreasing when $n_{-i}$ increases and $0<A_i\leq1$.
Then the total reward distributed to agents in layer $L_l$ in the budget distribution scheme is

\begin{equation}
\label{eq:scheme}
     B_{l} = (1-\beta)\sum_{i\in L_{l}} b_i'+\beta\sum_{i\in L_{l}} (1-\alpha)^{n_{-i}} b_i' 
\end{equation}

After deciding the discount factor $A_i$, we need to define the details of the distribution in each layer to complete the mechanism. We propose the distribution algorithm between two adjacent layers below, which performs well as we will show in the following sections. 

\begin{framed}
 \noindent\textbf{Distribution Algorithm between Two Adjacent Layers}
 
 \noindent\rule{\textwidth}{0.5pt}
 
 \noindent\textsc{Input:} the graph $G$ and $b_i'$ for each $i\in L_l$.
 
 \begin{enumerate}
     \item For each agent $i\in L_l$, set $r_i = i.V_b+i.V_h$, initialize $i.V_b = b_i'$ and $i.V_h = 0$.
     \item For each agent $j\in L_{l+1}$, initialize $b_j' = 0$.
     \item For each agent $j\in L_{l+1}$
     \begin{enumerate}
         \item Let $P$ be the set of agents in $L_{l}$ who propagate information to $j$.
         \item For each agent $i'\in P$:
         
         \begin{itemize}
             \item For each agent $i\in L_l\setminus\{i'\}$, set
             \begin{gather*}
                 i'.V_h\leftarrow i'.V_h + \alpha\beta\cdot i.V_b \\
                 b_j'\leftarrow b_j'+\alpha(1-\beta)\cdot i.V_b\\
                 i.V_b \leftarrow i.V_b-\alpha\cdot i.V_b
             \end{gather*}
         \end{itemize}
     \end{enumerate}
 \end{enumerate}
 
 \noindent\textsc{Output:} $r_i$ for each agent $i\in L_l$ and $b_j'$ for each agent $j\in L_{l+1}$.
\end{framed}


 
 
 
 
  
 
            
            
            
            
        
    

Intuitively, in the distribution algorithm above, for each agent $i\in L_l$, her reward includes two parts, the $V_h$ part and $V_b$ part where $V_b$ serves as the basic participation reward and $V_h$ serves as an extra reward for the propagation. If agent $i$ propagates information to another agent, all other agents in the same layer should provide a reward to agent $i$ and the new coming node. For example, considering an agent $j\in L_l$, it should give a total reward of $\alpha \cdot j.V_b$ to agent $i$, where $0<\alpha<1$. 
At last, we introduce $\beta$ to be the proportion of reward distribution between agent $i$ and her child. Thus, the new child in the next layer will get $\alpha \cdot (1-\beta) \cdot \sum_{j\in L_l, j\neq i}j.V_b$ and the agent $i$ will get $\alpha \cdot \beta \cdot \sum_{j\in L_l, j\neq i}j.V_b$. Notice that the final reward distribution is related to the order of agents in $L_{l+1}$ being selected in the step 3 in the distribution algorithm. Here we can just randomly choose an order. In practice, we can achieve more properties if the order is well-treated, which will be shown in Section~\ref{sec:time}. To see how a complete budget distribution with the above algorithm works, an example is illustrated below.

\begin{example}
Suppose the propagation network is shown in Figure~\ref{fig:ex2} and the budget is $\mathcal{B}=30$. In the mechanism, we set $\alpha = \beta = 0.2$. Now we show the process of the reward distribution.

\begin{figure}[H]
\centering
\includegraphics[width=4cm]{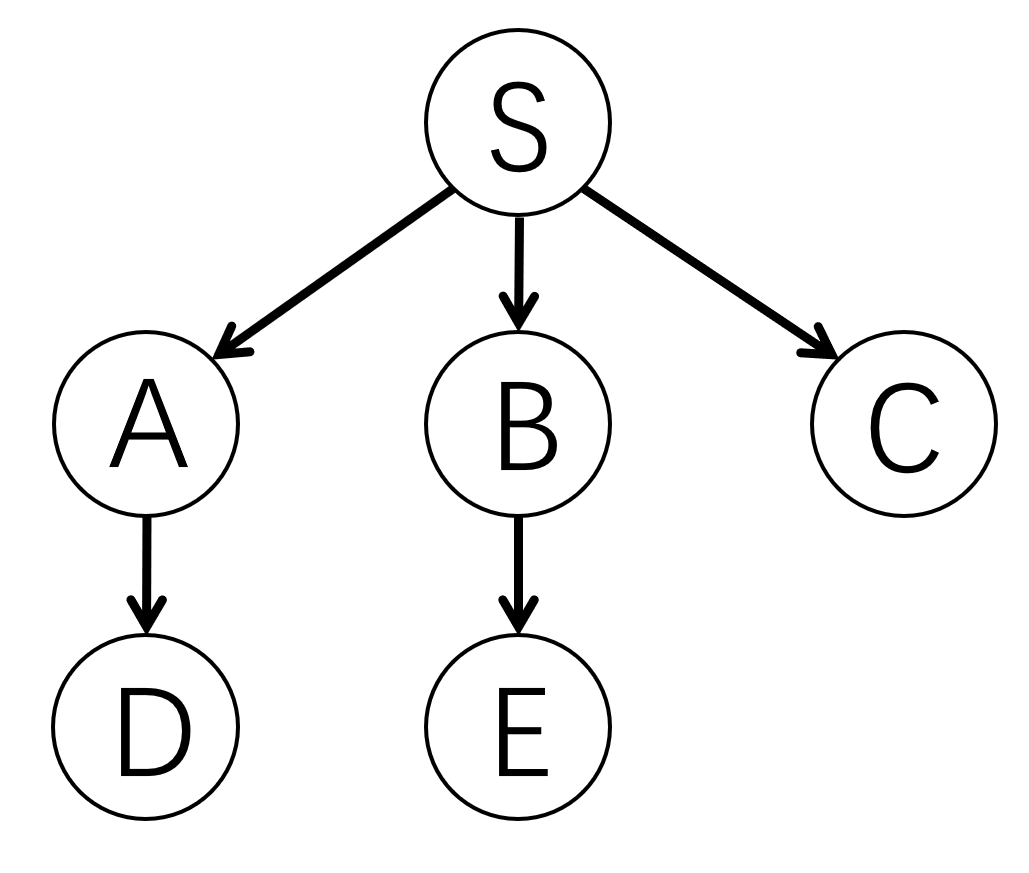}\
\caption{An example of the information propagation network}
\label{fig:ex2}
\end{figure}

\begin{itemize}
    \item Firstly, set $A.V_b=B.V_b=C.V_b=\frac{1}{3}\mathcal{B} = 10$, $D.V_b = E.V_b = 0$ and for all agents their $V_h$ are set to 0. 
    \item When A propagates to D, A and D will be rewarded, while B and C will be taken away a proportion of reward. Since $\alpha = 0.2$, Then: 
    \begin{gather*}
        {B}.V_b=B.V_b-\alpha  \cdot B.V_b = 8\\
        {C}.V_b=C.V_b-\alpha  \cdot C.V_b = 8
    \end{gather*}
    
    The reward taken from B and C will be distributed between A and D by a proportion of $\beta$, therefore:
    \begin{gather*}
        {A}.V_h=\beta  \cdot \alpha ({B}.V_b+{C}.V_b) = 0.8\\
        {D}.V_b=(1-\beta)  \cdot \alpha ({B}.V_b+{C}.V_b) = 3.2
    \end{gather*}
    \item Similarly, when B propagates to E, B and E will be rewarded. Notice that here $C.V_b$ has been updated:
    \begin{gather*}
    {B}.V_h = \beta \cdot \alpha (A.V_b + C.V_b) = 0.72\\
    {E}.V_b=E.V_b + (1-\beta) \cdot \alpha(A.V_b + C.V_b) = 2.88\\
    {A}.V_b=A.V_b-\alpha  \cdot A.V_b = 8\\
    {C}.V_b=C.V_b-\alpha \cdot C.V_b = 6.4
    \end{gather*}
\end{itemize}
\end{example}


At last, we show that the distribution algorithm satisfies the scheme described in Equation~\eqref{eq:scheme}. For each agent $i$ in layer $L_l$, the part $V_b$ of the final reward distributed to her is $(1-\alpha)^{n^{-i}} \cdot b_i'$. On the other hand, the total propagation reward to all agents in $L_l$ is $(1-\beta)\sum_{i\in L_l}[1-(1-\alpha)^{n_{-i}}]b_i'$. Therefore, the total final reward distributed to layer $L_l$ is:


\begin{displaymath}
    \begin{aligned}
    B_l &= \sum_{i\in L_l}(1-\alpha)^{n_{-i}} \cdot b_i' +  (1-\beta)\sum_{i\in L_l}[1-(1-\alpha)^{n_{-i}}]b_i' \\
    &= \sum_{i\in L_l}\{(1-\alpha)^{n_{-i}} \cdot b_i' +  (1-\beta)[1-(1-\alpha)^{n_{-i}}]\}b_i' \\
    &= \sum_{i\in L_l}[(1-\beta)+\beta(1-\alpha)^{n_{-i}}]b_i' \\
    &= (1-\beta)\sum_{i\in L_l} b_i'+\beta\sum_{i\in L_l}(1-\alpha)^{n_{-i}}b_i'
    \end{aligned}	
\end{displaymath}
which is identical to Equation~\eqref{eq:scheme}.


\subsection{Properties of the Budget Distribution Scheme}

In this section, we show that the budget distribution scheme with our distribution algorithm holds all the properties required in our model.

\begin{thm}\label{thm:main}
The budget distribution scheme that uses the distribution algorithm is individually rational, budget balanced and strongly propagation incentive compatible\footnote{Here we assume that for each layer we have more than one agent. For more general cases, we will discuss it in Section~\ref{sec:special}}.
\end{thm}

\begin{proof}
For IR and WBB, the proof is completed in Theorem~\ref{thm:bds}.

For budget balanced, when comes to the last layer $L_{l_{\text{max}}}$, for all agent $i\in L_{l_{\text{max}}}$, $n_i = 0$, so, $B_{l_{\text{max}}} = B_{l_{\text{max}}}'$ with $A_i=1$. Therefore, the mechanism is budget balanced.

For strongly PIC, it can be satisfied if for any agent $i$, we have
\begin{enumerate}
    \item When agent $i$ propagates information to a new agent, it will increase her reward.
    \item Whatever her descendants decide to act, it should not decrease the agent $i$'s reward.
\end{enumerate}

For the first condition, when an agent $i$ propagates information to a new agent, she gains rewards from other agents. Assume that the current reward owned by agent $i$ is $r_i$, then after the propagation, her reward will become
\begin{equation*}
    r'_i = r_i+\beta \cdot \alpha \sum_j r_{j}
\end{equation*}
where $j$ represents all the agents in the same layer of $i$'s other than $i$ herself. Since $\beta$, $\alpha$ and $r_{j}$ are always positive, then the second term should always be positive. Therefore $r'_i > r_i$, and the first condition is satisfied.

For the second condition, we need to focus on the propagation of an agent's descendant. Notice that in the algorithm, when one of the descendants propagates information to others, she can only take reward from the agents who do not belong to its ancestors.

Therefore, The budget distribution scheme that uses the distribution algorithm is individually rational, budget balanced and strongly propagation incentive compatible.
\end{proof}

Notice that Theorem~\ref{thm:main} proves not only the property of PIC, but also strongly PIC. This makes propagating information to all neighbours a unique dominant strategy for each agent. Therefore, in practice, the mechanism can achieve the maximal information propagation without the limitation of time.


\section{Practical Issues}
In this section, we will discuss some issues we may meet in practice.
\subsection{Special Cases of the Scheme}
\label{sec:special}
In the distribution algorithm between two adjacent layers discussed in previous section, if there is only one agent in the upper layer, then there will be no agents' reward to take from and the parent and the new agent will get nothing from the propagation. To solve this problem, we divide it to two different cases for an agent $i$ with only herself in the layer in our solution:
\begin{enumerate}[c{a}se 1:]

    \item There are leaf agents in previous layers that are not $i$'s ancestors.
    \item There is no leaf agents in previous layers that is not $i$'s ancestors.

\end{enumerate}

For case 1, we can simply treat those leaf agents as upper layer agents, i.e., we can take reward from those agents.
For case 2, which is much more complicated, the main idea is that tracing up to find where there are more than one agent in the layer and try to take reward from them. Notice that the reward we take away should be well-designed since it may reduce the ancestor agents' reward, which may harms the property of PIC.

More specifically, we present the solution to the special cases below\footnote{Here we assume that there are at least two agents in the first layer. In practise, the requires the sponsor has at least two neighbours, which is normal and reasonable.}.

\begin{framed}
 \noindent\textbf{Distribution Algorithm for Single Agent Layer}
 
 \noindent\rule{\textwidth}{0.5pt}
 
 \noindent\textsc{Input:} the graph $G$ and $b_i'$, where $L_l = \{i\}$.
 
 \begin{enumerate}
     \item For agent $i$, set $r_i' = i.V_b+i.V_h$ where $i.V_b = b_i'$ and $i.V_h = 0$.
     \item For each agent $j\in L_{l+1}$, initialize $b_j' = 0$.
     \item Let $F$ be the set of all leaf agents other than $i$ in $G_l$, where $G_l$ is the subgraph of $G$ with only agents in $\{S\}\cup L_1\cup\dots\cup L_l$.
     \item For each agent $j\in L_{l+1}$,
     \begin{enumerate}
         \item If $F\neq \emptyset$, then for each agent $f\in F$, set
         \begin{gather*}
             i.V_h\leftarrow i.V_h + \alpha\beta\cdot f.V_b\\
             b_j'\leftarrow b_j'+\alpha(1-\beta)\cdot f.V_b\\
             f.V_b \leftarrow f.V_b-\alpha\cdot f.V_b
         \end{gather*}
         \item Otherwise, set $x$ be the closest ancestor to $i$ which has more than one parent agents ($x$ must exists since $F=\emptyset$ and we assume $|L_1|>1$) and $q$ be the distance between $i$ and $x$. Then, for all agent $y$ in $x$'s parent agents, let $m$ be the number of $x$'s parents, set
             \begin{gather*}
                 b_j' \leftarrow b_j' + \frac{\alpha}{m\cdot 2^q} \cdot y.V_b \cdot (1 - \beta)\\
                 i.V_h \leftarrow i.V_h + \frac{\alpha}{m\cdot 2^q} \cdot y.V_b \cdot \beta\\
                 y.V_b \leftarrow y.V_b - \frac{\alpha}{m\cdot 2^q} \cdot y.V_b
             \end{gather*}
     \end{enumerate}
 \end{enumerate}
 
 \noindent\textsc{Output:} $r_i$ for agent $i$ and $b_j'$ for each agent $j\in L_{l+1}$.
\end{framed}

Here we provide an example of the special cases for better understanding. In Figure~\ref{fig:chouchou3}, assume that the sponsor $S$ owns a budget $\mathcal{B}=30$ and $\alpha = \frac{1}{5}, \beta = \frac{1}{5}$. As the initialization in the budget distribution scheme, the sponsor $S$ first gives out all the reward to agent A, B and D. So we have $A.V_b = B.V_b = D.V_b = 10$. Then all 3 agents of the first layer propagates information to C and, according to the budget distribution scheme and distribution algorithm, $C.V_b = 8.64$ and $A.V_b = B .V_b =D.V_b = 6.4$.

When agent C further propagates information to agent E, notice that there is only other agents in the same layer with C. Hence, we should apply distribution algorithm for single agent layer above. Since there is no leaf nodes in previous layers and the ancestor node C has three parent agents. Consequently, C and E will take reward from C's parents. Finally,
\begin{gather*}
    C.V_h = \frac{\alpha}{3\cdot 2}\cdot \sum_{i\in L_1} i.V_b\cdot \beta = 0.128\\
    E.V_b = \frac{\alpha}{3\cdot 2}\cdot \sum_{i\in L_1} i.V_b\cdot (1-\beta) = 0.512
\end{gather*}
and $A.V_b = B .V_b =D.V_b \approx 6.187$.


\begin{figure}[H]
\centering
\includegraphics[width=4cm]{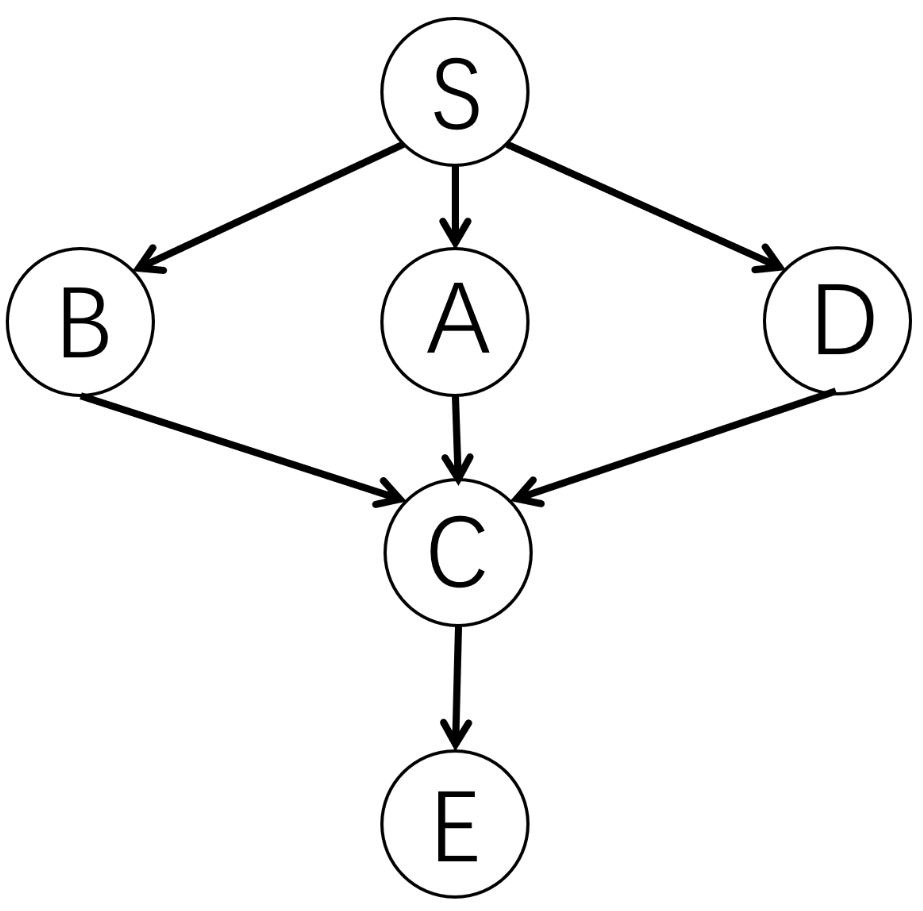}\
\caption{Example for the special case where there exists single agent layers.}
\label{fig:chouchou3}
\end{figure}

\begin{thm}
The budget distribution scheme that uses distribution algorithm for single agent layer is still strongly propagation incentive compatible, individual rationality and budget balanced.
\end{thm}
\begin{proof}
Here we have to prove the completeness of the special cases where there exists single agent layers. For individual rationality and budget balanced, it is obviously maintains in special cases. The critical part is for strongly propagation incentive compatibility.

For strongly PIC, with distribution algorithm for single agent layer, we must guarantee that the rewards taken from the parent agents should be less than the rewards taken if they choose not to propagate. 

There are two cases for single agent layers as described above. First, there are nonempty leaf agents set subpgraph with previous layers. In this case, these leaf agents have the same role as normal agents in upper layer in distribution algorithms. Thus from Theorem~\ref{thm:main} we know the strongly PIC still holds.

In the second case, we have no such leaf agents and have to get the reward from ancestor agents. Taking Figure:~\ref{fig:chouchou3} as an example, the latter agent may seek back to upward layers for the reward. Thus we need to consider if the propagation of the upper layer node which offers reward is truthful or not, in this case, the propagation from node A, B or D to C. Because each agent in a segment of a single chain may take reward from the upper layer node, we can derive the following statement:
\begin{equation*}
    \text{(Reward taken with propagation)}\leq \sum_{i=0}^{\infty}\frac{\alpha}{2^{i}m} \cdot b_{x}
    \leq \frac{2\alpha \cdot b_x}{m}
\end{equation*}
where $m$ again is the number of agents who are the original ancestors of the single chain. For example, in Figure:~\ref{fig:chouchou3}, the original ancestors of the single chain C-E are A, B and D, so $m=3$ in this case. In all circumstances $m$ should always be bigger than one because the termination condition in the mechanism is to find an ancestor node with more than one parent agent, and the ancestor node's parent number is $m$. Hence, we have:
\begin{equation*}
    \text{(Reward taken without propagation)} \geq \alpha \cdot b_x \geq \frac{2\alpha \cdot b_x}{m}
\end{equation*}

This indicates that propagation will help the node to lose less reward.

Finally, we can prove that it is strongly PIC in both special cases. The key idea is to use a prisoners' dilemma as shown in Table~\ref{tbl:PD}, where $\epsilon > 1$ and $r>0$. Each cell in the table represents the reward of the agent's action, where the first entry represents the reward of A and the second represents the reward of B. Along with the property of strongly PIC in budget distribution scheme with distribution algorithm, we can prove that the version with algorithm for single agent layer is still strongly PIC.
        
\centering
\begin{table}
\begin{tabular}{|c|c|c|}
\hline
 \diagbox[width=8em]{A}{B}& Propagate  & Not Propagate \\\hline
Propagate & ($-\alpha \cdot b_A / \epsilon$ , $-\alpha \cdot b_B / \epsilon$) & (+r,$-\alpha \cdot b_B$) \\\hline
Not propagate &  ($-\alpha \cdot b_A$ , +r) & (0 , 0) \\
\hline
\end{tabular}
\caption{Reward for Prisoners' Dilemma in our mechanism}
\label{tbl:PD}
\end{table}
\centering

\end{proof}

\subsection{Time Efficiency}
\label{sec:time}
At last, we consider the case what will happen in practice where we introduce a new concept of the time vector $\Vec{t}$ here. We discuss a new property called time efficiency. Participants should be willing to propagate as soon as possible in the mechanism; otherwise there may exist a deadlock in propagation. For example, if the mechanism is designed to make every node willing to propagate after a first propagation led by another node, none of the nodes will be willing to propagate first. This may be disastrous to the real-world applications of our mechanisms. Thus we define the property of time efficiency.

\begin{myDef}
A mechanism M is time efficient if and only if for any $G= (N,E)\in \mathcal{G}$, the budget $B\in \mathbb{R}_+$ and $\forall i\in N\setminus \{S\}$, for all $j \in \{j|(i,j)\in E\}$, changing $t'_j > t_j$ will result in 
\begin{equation*}
    M(G,B,t)_i \geq M(G,B,t')_i
\end{equation*}
\end{myDef}

\begin{proposition}
The budget distribution scheme that uses distribution algorithm and its variant for single agent layer is time efficient if we choose the agent with the order of their arriving time in the outer loop.
\end{proposition}

\begin{proof}
Consider the reward of node $i$, which has two parts $V_b$ and $V_h$, thus $r_i = i.V_{b}+i.V_{h}$. Let the reward of propagation obtained by node $i$ to be $r_i$. If $i$ propagates at time t, $ i.V_h = \alpha \cdot \beta \sum_{j\in J} j.V_{b}^{(t)}$, where J represents all the nodes that will be taken reward from because of $i$'s propagation. Now assume that $i$ choose to propagate later at $t'$, if there are another $g_j$ propagations during the time interval $[t,t']$ that will take reward from node $j$ where $g_j \geq 0$, $j.V_b$ will be $j.V_b^{(t')} = (1-\alpha)^{g_j}j.V_{b}^{(t)}$. So for all $j$, we have $j.V_b^{(t')} \leq j.V_{b}^{(t)}$, and we have:
\begin{equation*}
    \sum_j j.V_{b}^{(t)} \geq \sum_j j.V_{b}^{(t')} \quad  \textbf{for } t<t'
\end{equation*}
where $\sum V_{jb}^{(t)}$ is decreasing in terms of t. Therefore spreading immediately is the dominant strategy.
\end{proof}


\section{Conclusions}
In this paper, we formalized an information propagation game on a network where a sponsor is willing to pay a fixed reward to incentivize agents to propagate information for her. The model has many promising applications where social networks are well engaged, such as viral marketing and questionnaire survey. We proposed a novel scheme of the reward sharing mechanism, which is propagation incentive compatible and budget balanced. We also offered an instance of our scheme where strongly PIC, BB and time efficiency are achieved.

There are also several aspects that deserve further investigation. For example, the Sybil attack is a common issue in real-world applications based on social networks. It is quite challenging to achieve both PIC and Sybil proof in our settings. One other possible improvement is to extend our solutions to the settings of resource or task allocation, where we have the dimension of competition for the resource or task. 


\bibliography{ecai}

\end{document}